\documentclass[11pt,a4paper]{article}
\usepackage{amsmath,amssymb,amsfonts,latexsym,graphicx,amsthm}
\usepackage{thmtools}
\usepackage{color}
\usepackage{url,hyperref}
\usepackage{comment}
\usepackage[linesnumbered,boxed,ruled,noend]{algorithm2e}
\usepackage{xcolor} \usepackage{hyperref} \hypersetup{
    colorlinks = true, allcolors = black, citecolor = {blue!0!black}, linkcolor = {blue!0!black}
}
\usepackage[shortlabels]{enumitem}
\usepackage{cleveref}

\usepackage[margin=1in]{geometry}

\newtheorem{theorem}{Theorem}[section]
\newtheorem{lemma}[theorem]{Lemma}
\newtheorem{claim}[theorem]{Claim}

\theoremstyle{definition}
\newtheorem{definition}[theorem]{Definition}

\newtheorem{invariant}[theorem]{Invariant}
\newtheorem{question}[theorem]{Question}

        {\medskip}

\newcommand{\eps}{\epsilon}

\newcommand{\ceil}[1]{\left\lceil #1 \right\rceil}
\newcommand{\floor}[1]{\left\lfloor #1 \right\rfloor}

\newcommand{\brac}[1]{\left(#1\right)}
\newcommand{\indic}[1]{\mathbf{1}\left[#1\right]}

\newcommand{\cnt}{\mathsf{cnt}}

\newcounter{paragraphCounter} % Create a new counter

\usepackage{tikz}
\usetikzlibrary{trees}
\usetikzlibrary{shapes.geometric}

\begin{document}

\title{Faster Deterministic Streaming Vertex Coloring}

\author{
	Shiri Chechik\thanks{Tel Aviv University, \href{}{shiri.chechik@gmail.com}} \and 
	Hongyi Chen\thanks{State Key Laboratory for Novel Software Technology, Nanjing University, \href{}{hongyi.chen@smail.nju.edu.cn}}\and    
	Tianyi Zhang\thanks{State Key Laboratory for Novel Software Technology, Nanjing University, \href{}{tianyiz25@nju.edu.cn}}
}

\date{}

\maketitle

\begin{abstract}
Graph coloring is a fundamental problem in computer science. In the semi-streaming model, an input graph \(G\) on \(n\) vertices and maximum degree \(\Delta\) is presented as a stream of edges, and the goal is to compute a vertex coloring using a small number of colors while storing only \(\tilde{O}(n)\) bits of memory.

Recent work has revealed an exponential separation between randomized and deterministic approaches in this setting: while randomized algorithms can achieve a \((\Delta+1)\)-coloring in a single pass [Assadi, Chen, and Khanna, 2019], any single-pass deterministic algorithm requires \(\exp(\Delta^{\Omega(1)})\) colors [Assadi, Chen, and Sun, 2022]. Consequently, deterministic algorithms that use few colors must necessarily make multiple passes over the stream. Prior to this work, the best known deterministic trade-offs were: an \(O(\Delta^2)\)-coloring in 2 passes, an \(O(\Delta)\)-coloring in \(O(\log \Delta)\) passes [Assadi, Chen, and Sun, 2022], and a \((\Delta+1)\)-coloring in \(O(\log \Delta \cdot \log\log \Delta)\) passes [Assadi, Chakrabarti, Ghosh, and Stoeckl, 2023]. It remained open whether better trade-offs—particularly with sub-logarithmic pass complexity and linear-in-\(\Delta\) palette size—were achievable.

In this paper, we present a new deterministic semi-streaming algorithm that computes an \(O(\Delta)\)-coloring in \(O\!\left(\sqrt{\log \Delta}\right)\) passes. This is the first deterministic streaming algorithm to achieve a coloring with palette size linear-in-\(\Delta\) using sublogarithmic-in-\(\Delta\) passes.
\end{abstract}

\thispagestyle{empty}
\clearpage
\setcounter{page}{1}

\section{Introduction}

Graph coloring is a fundamental problem in graph theory and computer science. Given an $n$-vertex undirected graph $G = (V, E)$ with a maximum degree of $\Delta$, the goal is to assign each vertex a color such that no edge connects two vertices of the same color. Beyond its combinatorial interest, graph coloring has many applications, including scheduling, register allocation, and resource assignment.

In this work, we focus on vertex coloring in the \emph{semi-streaming} model, which is designed for processing massive graphs. In this model, the edges of the input graph arrive sequentially as a stream, the algorithm can make only a limited number of passes over the stream, and the allowed working memory is constrained to $\tilde{O}(n)$ bits\footnote{$\tilde{O}$ hides poly-logarithmic factors.}. The semi-streaming model has inspired numerous algorithmic developments and lower bound results for fundamental graph problems. Our goal is to design streaming algorithms that, within these constraints, can output a proper vertex coloring while keeping the palette size (the total number of colors) and the number of passes to a minimum.

Similar to the classical sequential model, obtaining an optimal or near-optimal coloring in the streaming model is infeasible. Consequently, previous work on graph streaming algorithms has focused on obtaining colorings whose size is bounded by structural parameters of the input graph, such as the maximum degree $\Delta$, degeneracy, or arboricity. This question was first studied in \cite{10.5555/3310435.3310483}, where the authors presented a breakthrough randomized semi-streaming algorithm that computes a $(\Delta + 1)$-coloring in a single pass over the graph stream. Additionally, a contemporary result in \cite{bera2018coloringgraphstreams} achieved a randomized one-pass $O(\Delta)$-coloring, followed by subsequent work in \cite{bera_et_al:LIPIcs.ICALP.2020.11} that provided a $(1 + o(1))\kappa$-coloring, where $\kappa$ represents the degeneracy of graph $G$. Later, the study of randomized streaming vertex coloring culminated in \cite{10.1145/3519935.3520005} with a one-pass $\Delta$-coloring algorithm.

All the algorithms mentioned above are randomized. In sharp contrast, it was proven in \cite{assadi2022deterministic} that any single-pass deterministic algorithm requires $\exp(\Delta^{\Omega(1)})$ colors, indicating that multiple passes are essential for deterministic streaming algorithms that utilize a small number of colors, such as $\mathrm{poly}(\Delta)$ colors. As a complementary result, the authors of \cite{assadi2022deterministic} further demonstrated that allowing multiple passes enables deterministic algorithms to achieve significantly better performance, including $O(\Delta^2)$-colorings in 2 passes or $O(\Delta)$-colorings in $O(\log \Delta)$ passes. In a follow-up work \cite{assadi2023coloring}, the authors improved the color count from $O(\Delta)$ to $\Delta + 1$ by slightly increasing the number of passes to $O(\log \Delta \log \log \Delta)$.

Notably, there are currently no established lower bounds against deterministic multi-pass algorithms, leaving it unclear what the best vertex colorings computable in 2 passes are. Conversely, the question of the best pass complexity for \(O(\Delta) \)-colorings (or even \( (\Delta + 1) \)-colorings) remains open. Below, we would like to highlight the following question:

\begin{question}\label{quest}

What is the smallest number of passes required for computing \( O(\Delta) \)-colorings deterministically in the semi-streaming setting?

\end{question}

\subsection{Our Result}

Our main result is a deterministic semi-streaming algorithm that achieves an $O(\Delta)$-coloring of a graph in $O\brac{\sqrt{\log \Delta}}$ passes, which circumvents the logarithmic pass complexities in \cite{assadi2022deterministic,assadi2023coloring} and makes progress towards a better understanding of \Cref{quest}.

\begin{theorem}\label{theorem:main}
Let $G = (V, E)$ be an input graph on $n$ vertices with maximum degree $\Delta$ accessible via a data stream. For any constant $\eta\in (0, 1)$, there is a deterministic streaming algorithm that finds an $(1+\eta)\Delta$-coloring of $G$ using $O\brac{\frac{1}{\eta}\sqrt{\log \Delta}}$ passes and $\tilde{O}(n)$ bits of space.
\end{theorem}

Compared to $O(\Delta)$-coloring studied in \cite{assadi2022deterministic}, we have optimized the constant coefficient in front of $\Delta$ in the number of colors from $O(1)$ to $1+\eta$ for any constant parameter $\eta\in (0,1)$, but this is not a major improvement and it comes with a minor adjustment of a subroutine in \cite{assadi2022deterministic}.

\subsection{Related Work}

Ben-Eliezer et al.~\cite{ben2022framework} introduced the adversarially robust streaming model which lies conceptually between deterministic and randomized streaming algorithms.
In this model, the adversary can adaptively generate the stream: at any time, it may request the algorithm to produce an output and then choose the next input based on the entire history (the transcript of past inputs and outputs). The algorithm must produce correct outputs with high probability, at all times and against any adaptive adversary.

Graph coloring in the adversarially robust streaming model was first studied by Chakrabarti et al.~\cite{chakrabarti_et_al:LIPIcs.ITCS.2022.37}, who established a one-pass lower bound: any valid $K$-coloring algorithm must use at least $\Omega(n\Delta^2 / K)$ bits of space. In particular, a semi-streaming algorithm would require $\Omega(\Delta^2)$ colors. On the algorithmic side, the same work presented an $O(\Delta^2)$-coloring algorithm using $O(n\sqrt{\Delta})$ space and an $O(\Delta^3)$-coloring algorithm using $O(n)$ space, both assuming oracle access to $O(n\Delta)$ random bits. Subsequently, Assadi et al.~\cite{assadi2023coloring} improved these results, showing that an $O(\Delta^{5/2})$-coloring can be achieved with oracle access to $O(n\Delta)$ random bits, and that an adversarially robust $O(\Delta^3)$-coloring is possible in semi-streaming space even when the randomness used by the algorithm is counted toward its space usage.
Despite these advances, a polynomial gap remains between the best known lower and upper bounds in the adversarially robust setting.

\subsection{Technical Overview}

\paragraph*{Previous Approach.} 
% We begin by reviewing the method of~\cite{assadi2023coloring}.  
% The algorithm proceeds in a sequence of epochs. Each epoch starts with a partial coloring~$\chi$ in which a subset~$U \subseteq V$ remains uncolored. The goal of the epoch is to extend~$\chi$ by coloring at least one third of the vertices in~$U$, thereby reducing the size of~$U$ to at most~$\frac{2}{3}|U|$. Initially,~$U = V$. After at most~$\lceil \log_{3/2} \Delta \rceil$ epochs, the uncolored set satisfies~$|U| \le n/\Delta$. At this point, the algorithm performs a final pass to gather all edges incident to~$U$ and completes the coloring greedily.  
% The $i$-th epoch requires~$O(\log \Delta / i)$ passes, so the total number of passes is bounded by~$O(\log \Delta \log \log \Delta)$.

We begin by reviewing the approach of~\cite{assadi2022deterministic}.  
Their algorithm maintains a partial coloring $\phi: V\rightarrow [O(\Delta)]\cup \{\bot\}$ ($\bot$ means uncolored) of the graph $G$ and gradually extends it to reduce the number of uncolored vertices over multiple passes.
Initially, it starts with an empty coloring where all vertices are uncolored. Then, they show that in every $O(1)$ passes over the input stream, one could reduce the number of uncolored vertices by a constant factor. So, in the end, the total number of passes would be $O(\log\Delta)$. Morally speaking, imagine that we could use randomness and in each round we draw a uniformly random color $c(v)$ from $O(\Delta)$ to each uncolored vertex $v\in \phi^{-1}(\bot)$. Then, we can argue that for each vertex $v\in \phi^{-1}(\bot)$, the probability that $c(v)$ conflicts with some neighboring $c(u)$ or $\phi(w)$ is bounded away from $1$, or more precisely:
\begin{equation}\label{equ:conflict}
\Pr[\exists u\in N(v)\cap \phi^{-1}(\bot), c(v) = c(u)] + \Pr[\exists w\in N(v)\setminus \phi^{-1}(\bot), c(v) = \phi(w)] < 1 - \Omega(1)
\end{equation}
Then, we can extend the partial coloring $\phi$ by assigning $c(v)$ to each uncolored vertex $v$ without any conflicts with its neighbors, which would reduce the number of uncolored vertices by a constant factor. Derandomizing this approach can be done using universal hash functions instead of independent random colors; this is because we only need to upper bound the total number of vertex conflicts, and so we can choose the hash function with the minimum number of conflicts to extend $\phi$.

\paragraph*{Some Natural Attempts.} To reduce the number of passes, a natural observation is that the $v$'s conflict probability (the first term in \Cref{equ:conflict}): $$\Pr[\exists u\in N(v)\cap \phi^{-1}(\bot), c(v) = c(u)]$$ becomes smaller when the number of uncolored vertices shrinks. To see this, imagine that at the moment $|\phi^{-1}(\bot)| = \eps n$ for some $\eps\in (0, 1)$. Then hopefully each vertex $v\in \phi^{-1}(\bot)$ has only $\epsilon \Delta$ neighbors in vertex set $\phi^{-1}(\bot)$ (proportional to the size of $\phi^{-1}(\bot)$ compared to $V$). If we could draw a uniformly random color for each $v\in \phi^{-1}(\bot)$ from all $v$'s available colors $[O(\Delta)]\setminus \{\phi(u)\mid u\in N(v)\}$, then we could completely avoid the second conflict probability term in \Cref{equ:conflict}:
$$\Pr[\exists w\in N(v)\setminus \phi^{-1}(\bot), c(v) = \phi(w)]$$
and then the conflict probability of $v$ would be at most:
$$\Pr[\exists u\in N(v)\cap \phi^{-1}(\bot), c(v) = c(u)]\leq \epsilon \Delta \cdot O\brac{\frac{1}{\Delta}} = O(\epsilon)$$
So, under a random color extension (which always draws available colors) to $\phi$, in expectation a proportion of $1 - \Omega(\epsilon)$ would be colored, leaving $O(\epsilon^2n)$ uncolored vertices; in other words, the ratio $|\phi^{-1}(\bot)| / n$ would drop quadratically. Therefore, the total number of passes would be at most $O(\log\log\Delta)$.

The main difficulty of implementing this natural attempt is that we do not have direct access to the available colors of each vertex $v\in \phi^{-1}(\bot)$ which is $[O(\Delta)]\setminus \{\phi(u)\mid u\in N(v)\}$; that is, those colors not yet used by any of $v$'s neighbor under color assignment $\phi$. Instead, what we do have access to is only a data stream of unavailable colors $\{\phi(u)\mid u\in N(v)\}$ for each $v\in \phi^{-1}(\bot)$ converted from the input graph stream. This task of accessing available colors given a data stream of unavailable colors can be viewed as a \emph{missing item} problem. In the standard formulation of the missing item problem, the algorithm is given a stream of integers $x_1, x_2, \ldots x_N\in [M]$, and the goal is usually to find some or (approximately) count all elements in $[M]\setminus \{x_1, x_2, \ldots, x_N\}$. This problem was studied in \cite{DBLP:conf/soda/Stoeckl23} and the author showed a deterministic space lower bound of $\tilde{\Omega}(N)$ against any streaming algorithm that outputs any single missing item. In our case, $N$ is as large as $\Delta$, and so we would need to allocate $\tilde{\Omega}(\Delta)$ space for each $v\in \phi^{-1}(\bot)$ which is infeasible.

\paragraph*{Vertex and Color Partitions.}
As hinted by the impossibility result from \cite{DBLP:conf/soda/Stoeckl23}, it is generally hard to assume access to the set of available colors $[O(\Delta)]\setminus \{\phi(u)\mid u\in N(v)\}$ and bypass the conflicts with neighbors that are already colored, which is the main source of technical difficulty, without using space of palette size $O(\Delta)$, per uncolored vertex $v$.

To circumvent the $O(\Delta)$ space requirement per uncolored vertex while still avoiding the conflicts with colored neighbors, the basic idea is to reduce the palette size. More concretely, we are going to partition the palette into much smaller sub-palettes:
$$[O(\Delta)] = C_1\cup C_2\cdots \cup C_l$$
and also partition the uncolored vertex set into vertex subsets:
$$\phi^{-1}(\bot) = B_1\cup B_2\cdots \cup B_l$$
such that it is possible to assign colors from $C_i$ to all vertices in $B_i$ in a valid manner, $\forall 1\leq i\leq l$. If all the sub-palettes $C_i$ are small, and the union of the induced subgraphs $G[B_i]$ fits in $\tilde{O}(n)$ bits of memory, then we can compute and store all the available colors and the subgraphs in memory over a single pass, and then complete the color extension.

This vertex and color partitioning approach has been previously adopted in some recent distributed graph coloring literature \cite{DBLP:conf/podc/ChangFGUZ19,DBLP:conf/icalp/Parter18,DBLP:conf/wdag/ParterS18,czumaj2020simple,DBLP:conf/icalp/CoyCDM23}, in which we found \cite{czumaj2020simple} particularly helpful for our reference. The authors of \cite{czumaj2020simple} worked on the problem of computing a $(\Delta+1)$-coloring in the congested clique model. Their main algorithm can deterministically compute a pair of vertex and color partitions but with an extra bad vertex subset:
$$\begin{aligned}
	V &= B_1\cup B_2\cdots \cup B_{l = \Delta^{0.9}}\cup B_{\text{bad}}\\
	[\Delta+1] &= C_1\cup C_2\cdots \cup C_{l = \Delta^{0.9}}
\end{aligned}$$
such that for every vertex $v\in B_i, \forall 1\leq i\leq l$, the size of $C_i$ is strictly larger than the number of its neighbors in $B_i$, or more formally $\left| C_i \right| > \deg_{G[B_i]}(v)$ (note that they directly compute a full coloring in one shot after partitioning, and they do not repeatedly use color extensions like us, so they do not need to exclude any unavailable colors from sub-palette $C_i$). Here the set $B_{\text{bad}}$ will be a small vertex set of size $O(n / \Delta)$ collecting all the vertices that violate this inequality. In the end, they show that each induced subgraph $G[B_{\text{bad}}], G[B_i], \forall 1\leq i\leq l$ has size $O(n)$ which fits in the local memory of a single machine, so they could compute the whole coloring using $O(n)$ machines in $O(1)$ communication rounds.

\paragraph*{Deterministic Partitions via Almost $k$-Wise Independence.} To deterministically compute a good pair of vertex and color partitions, the authors of \cite{czumaj2020simple} used $O(1)$-wise independent hash families to map vertices and colors to their subsets and sub-palettes. More specifically, they build two $O(1)$-wise independent hash families $\mathcal{H}_1, \mathcal{H}_2$ where $\mathcal{H}_1$ contains functions $h: V\rightarrow [\Delta^{0.9}]$ and $\mathcal{H}_2$ contains functions $h: [\Delta+1]\rightarrow [\Delta^{0.9}]$; in reality, the exponent of range is a much smaller constant (say, $0.1$ instead of $0.9$), and the actual algorithm needs to repeat the partitions for $O(1)$ times \cite{czumaj2020simple}. Then, they used the method of conditional expectation to find a good pair of hash functions $(h_1^*, h_2^*)\in \mathcal{H}_1\times \mathcal{H}_2$ which is used as their vertex and color partitions.

It turns out that in our streaming setting, as opposed to the congested clique setting, it is rather prohibitive to find a good hash function out of an $O(1)$-wise independent hash family. According to standard literature \cite{DBLP:journals/siamcomp/PaghP08,DBLP:journals/jal/AlonBI86,DBLP:conf/focs/ChorGHFRS85}, $\mathcal{H}_1, \mathcal{H}_2$ have sizes polynomial in terms of their domains which are $n^{O(1)}, \Delta^{O(1)}$, respectively. To decide if a hash function pair $(h_1, h_2)\in \mathcal{H}_1\times \mathcal{H}_2$ induces a good partition for any specific vertex $v\in \phi^{-1}(\bot)$, we need to count its degree in the same bin which is $\left| \{h_1(v) = h_1(u)\mid u\in N(v) \}\right|$, as well as estimate its available colors which is $\{c\mid h_2(c) = h_1(v)\} \setminus \{\phi(u)\mid u\in N(v)\}$. This requires maintaining a counter for each pair $(h_1, h_2)$ and each vertex $v\in \phi^{-1}(\bot)$ during one stream pass, and this totals a space usage of $O\brac{|\phi^{-1}(\bot)|\cdot \brac{n^{O(1)} + \Delta^{O(1)}}} = n^{1+O(1)}$.

In order to reduce the space, instead of relying on $O(1)$-wise independence, we will switch to almost $O(1)$-wise independence \cite{alon1992simple,DBLP:journals/siamcomp/NaorN93} where the hash family only has size roughly $M^{O(1)}\log^{O(1)}N$ with $N, M$ being the domain and range sizes. Plugging in our setting, $N = n$ so the memory requirement would be $\tilde{O}\brac{|\phi^{-1}(\bot)|\cdot M^{O(1)}}$. Still, we cannot use the same domain size of $M = \Delta^{\Omega(1)}$ as in \cite{czumaj2020simple}. This leads to the following modifications in algorithm parameters.
\begin{itemize}
	\item To make $|\phi^{-1}(\bot)|$ moderately small, we can run the color extension method of~\cite{assadi2022deterministic} for $\sqrt{\log \Delta}$ rounds as a preprocessing step. This takes $O(\sqrt{\log\Delta})$ passes over the data stream and brings down the number of uncolored vertices to $|\phi^{-1}(\bot)|\leq n / 2^{\Omega(\sqrt{\log\Delta})}$.
	
	\item Set the range size of hash functions $M = 2^{\Theta\brac{{\sqrt{\log\Delta}}}}$ which keeps the total memory requirement below $\tilde{O}(n)$. Since in the end we want a color partition of small size so that we can deterministically compute all the available colors around each uncolored vertex, we will recursively partition each sub-palette for $\log_M\Delta = O(\sqrt{\log\Delta})$ rounds.
\end{itemize}

\paragraph*{Pruning Frequent Neighbor Colors.} Let us focus on a single round where we are given a pair of partitions:
$$\begin{aligned}
	\phi^{-1}(\bot) &= B_1\cup B_2\cdots \cup B_{l}\cup B_{\text{bad}}\\
	[O(\Delta)] &= C_1\cup C_2\cdots \cup C_{l}
\end{aligned}$$
and our goal is to refine this partition further so that the sizes of sub-palettes $C_i$ and subgraphs $G[B_i]$ become smaller by a factor of $2^{-\Omega(\sqrt{\log\Delta})}$. As discussed above, we will use almost $O(1)$-wise independent hash families $\mathcal{H}_1, \mathcal{H}_2$ that map uncolored vertices and colors to the range $[M]$ for some $M = 2^{\Omega(\sqrt{\log\Delta})}$. Each vertex $v\in B_i$ will compute some statistics during a single pass over the stream for each pair of hash functions $(h_1, h_2)\in \mathcal{H}_1\times \mathcal{H}_2$:
\begin{enumerate}[(i)]
	\item the number of neighbors under $h_1$, namely the size of the set $\{u\mid h_1(u) = h_1(v), u\in B_i\}$;
	\item the (approximate) number of available colors in the color set: $$\{c\mid h_2(c) = h_1(v), c\in C_i\}\setminus \{\phi(u)\mid u\in N(v)\}$$
\end{enumerate}
If the value of (ii) is larger than (i), then $v$ would give one vote to the pair $(h_1, h_2)$. In the end, we hope that there exists a good pair $(h_1^*, h_2^*)$ which receives votes from a proportion of at least $1 - 2^{-\Omega(\sqrt{\log\Delta})}$ uncolored vertices. If so, we would take this pair $(h_1^*, h_2^*)$ to generate a partition refinement, and all vertices which did not vote to $(h_1^*, h_2^*)$ would move to $B_{\text{bad}}$. After $O(\sqrt{\log\Delta})$ rounds of partition refinement, we will be able to color all vertices in $B_i$'s using colors from $C_i$'s in one pass, as $C_i$'s will be small enough, leaving a very small uncolored vertex set $B_\text{bad}$.

Assume inductively that $v$ has more available colors in $C_i$ than its neighbors in $B_i$. To preserve this property for a refined partition, we argue that both (i)(ii) will concentrate around their expectations using concentration inequalities for almost $k$-wise independence; such concentration bounds were previously studied in \cite{gravin2021concentration} but in the context of cryptography with a somewhat different definition of almost $k$-wise independence from ours, 
so we will verify our version of concentration (see \Cref{lem:conc}) in the appendix.

By the concentration bound, we know that for most vertices $v\in B_i$, it should vote for most pairs $(h_1, h_2)$. The main difficulty arises from verifying which pairs $(h_1, h_2)$ are good. Computing (i) is straightforward, whereas computing (ii) is much more challenging because it is still a missing item problem which is hard according to \cite{DBLP:conf/soda/Stoeckl23}. The difference is that here we only need a lower bound estimation of the number of the available colors, not output any specific available color. A basic idea is to lower bound (ii) by:
\begin{align}
	&\left|\{c\mid h_2(c) = h_1(v), c\in C_i\}\setminus \{\phi(u)\mid u\in N(v)\}\right|\\
	= &\left|\{c\mid h_2(c) = h_1(v), c\in C_i\}\right| -  \left|\{\phi(u)\mid u\in N(v), h_2(\phi(u)) = h_1(v), \phi(u)\in C_i\}\right| \label{ineq:distinct}\\
	\geq & \left|\{c\mid h_2(c) = h_1(v), c\in C_i\}\right| -  \left|\{u\in N(v)\mid h_2(\phi(u)) = h_1(v), \phi(u)\in C_i\}\right| \label{ineq:linear-est}
\end{align}
The second term (which will be our main focus) $\left|\{\phi(u)\mid u\in N(v), h_2(\phi(u)) = h_1(v), \phi(u)\in C_i\}\right|$ counts the number of distinct colors used by $v$'s neighbors, while the set $\{u\in N(v)\mid h_2(\phi(u)) = h_1(v), \phi(u)\in C_i\}$ is what the algorithm sees in the data stream. Ideally, if we could directly compute the number of distinct neighbor colors, then we could use \Cref{ineq:distinct} as our lower bound which aligns well with concentration inequalities. However, this is impossible for deterministic algorithms without using $\tilde{\Omega}(\Delta)$ space \cite{DBLP:conf/focs/IndykW03}, and we could only use the linear estimator \Cref{ineq:linear-est} as our lower bound; this linear estimator was also used in \cite{assadi2023coloring} for deterministic streaming $(\Delta+1)$-coloring.

Unfortunately, concentration bounds of the linear estimation $$\left|\{u\in N(v)\mid h_2(\phi(u)) = h_1(v), \phi(u)\in C_i\}\right|$$ of distinct neighbor colors could be very poor. Imagine that all colored neighbors of $v$ are only using $M/2$ different colors of $C_i$, with all these colors assigned evenly. Then, every such color itself would destroy a single vertex bin, and so the right-hand side of \Cref{ineq:linear-est} would be negative for nearly half of $h_1$. To fix this issue, we will first prune all those frequent neighbor colors away in a preprocessing step using a folklore deterministic algorithm for finding frequent items from \cite{10.5555/867576}. Then, we show that the concentration bounds could work after such pruning steps.

\section{Preliminaries}
\iffalse
\begin{lemma}[Theorem B.1 of \cite{gravin2021concentration}]
Suppose $X_1, X_2, \cdots, X_n$ are random variables taking values in $\{0, 1\}$, and $a_1, a_2, \cdots a_n$ are approximate expectations of these variables. If for any subset $I \subseteq [n]$ with $|I| \le k$, we have 
$
\mathbf E \left[ \prod_{i \in I} \left( X_i - a_i \right) \right] \le \varepsilon^{|I|}
$. 
Then, 
$$
\Pr\left[\sum_{i=1}^n X_i \ge \sum_{i = 1}^n a_i + \gamma n \right] \le \sqrt2  \left( \frac{\varepsilon + \sqrt{\frac{k}{en}}}{\gamma}\right)^k
$$

\end{lemma}
\fi

\subsection{Streaming Vertex Coloring}

All logarithms in the main text will take base $2$ by default when we do not specify the base. For each positive value $x\geq 1$, let $\floor{x}_2$ be the largest integer power of $2$ less or equal to $x$. The input graph is denoted by $G = (V, E)$ on $n$ vertices and maximum degree $\Delta$. For each vertex $v\in V$, let $N(v)$ be the set of neighbors around $v$ in $G$.

We are working in the semi-streaming model where the edges of $G$ are read one-by-one in an arbitrary order, and the algorithm has $\tilde{O}(n)$ bits of space. Our goal is to compute a $(1+\eta)\Delta$-coloring of $G$ for any fixed constant $\eta\in (0, 1)$ (without loss of generality, let us assume $\eta\Delta$ is an integer by rounding downward).

\begin{definition}[partial coloring]
A \emph{partial coloring} of a graph $G=(V,E)$ is a mapping $\phi : V \to [(1+\eta)\Delta] \cup \{\perp\}$, where $\phi(v) = \perp$ indicates that vertex $v$ is currently uncolored. We require that for every edge $(u,v) \in E$ with both endpoints colored, $\phi(u) \neq \phi(v)$, i.e., the coloring induces no monochromatic edge among colored vertices. For an uncolored vertex $v\in V$ and a color $c$, we say $c$ is \emph{available} to $v$ if $c$ is not assigned to any of $v$'s neighbors under $\phi$.
\end{definition}

\begin{definition}[$(\deg+1)$-list coloring]
    In a $(\deg+1)$-list coloring problem, for each vertex $v$ in the input graph $G$, we are given a list $L_x$ of at least $\deg_G(x)+1$ available colors for $v$, and the goal is to compute a proper vertex coloring of $G$ such that each vertex $v$ chooses one color from its own list $L_v$.
\end{definition}
It is a folklore that a $(\deg+1)$-list coloring can be found in a greedy manner once the whole input is stored in memory.

\iffalse
\begin{definition}[extension of a partial coloring]
Given a partial coloring $\phi$, an \emph{extension} $\phi'$ of $C$ is another partial coloring satisfying:
\begin{itemize}
    \item For every vertex $v$ with $\phi(v) \neq \perp$, we have $\phi'(v) = \phi(v)$.
    \item For every vertex $v$ with $\phi(v) = \perp$, we have $\phi'(v) \in [c] \cup \{\perp\}$.
\end{itemize}
\end{definition}
\fi

\begin{restatable}[\cite{assadi2022deterministic}]{lemma}{extendcolor}
\label{lemma:prep}
Let $G$ be an $n$-vertex graph with maximum degree $\Delta$, presented in an insertion-only stream, and let $\phi$ be a partial coloring using at most $(1+\eta)\Delta$ colors.  There exists a deterministic algorithm that, using $O(1)$ passes and $\tilde{O}(n)$ bits of space, extends $\phi$ to a new coloring $\phi'$ such that 
$|\phi'^{-1}(\perp)| < \brac{1 - \Omega(\eta)} |\phi^{-1}(\perp)|$; in other words, the number of uncolored vertices drops by a $1 - \Omega(\eta)$ factor.
\end{restatable}

The original statement in \cite{assadi2022deterministic} only focused on $6\Delta$-coloring and did not optimize the constant in front of $\Delta$, so we will verify this statement 
in the appendix.

\subsection{Streaming Frequent Items}
We will also apply a simple deterministic algorithm for finding frequent items in a data stream which originally appeared in \cite{10.5555/867576}. %\hongyi{Misra–Gries heavy hitters algorithm? Reference: Finding repeated elements. Misra, David Gries}

\begin{lemma}[\cite{cmu15451lec06,10.5555/867576}]\label{lem:freq}
    Given a data stream $x_1, x_2, \ldots, x_N\in [M]$. Then, for any positive integer $k\geq 1$, there is a deterministic algorithm using $O(k\log N)$ space outputs a list of $\ceil{s}$ items $y_1, y_2, \ldots, y_k\in [M]$ such that any integer $y$ which appears in the stream more than $N/k$ times is included in the list $\{y_1, y_2, \ldots, y_{k}\}$.
\end{lemma}

\subsection{Almost $k$-Wise Independence}

\begin{definition}[almost $k$-wise independence]
    A sequence of random variables $X_1, X_2, \ldots, X_N\in [M]$ is $(k, \epsilon)$-wise independent, if for any $k$ different indices $1\leq i_1, i_2, \ldots, i_k\leq N$ and $k$ values $y_1, y_2, \ldots, y_k\in [M]$, we have:
    $$\left|\Pr\left[X_{i_j} = y_j, \forall 1\leq j\leq k\right] - \frac{1}{M^k}\right|\leq \epsilon$$
\end{definition}

\begin{definition}[almost $k$-wise independent hash functions]
    A family $\mathcal{H}$ of functions $h: [N]\rightarrow [M]$ is $(k, \epsilon)$-wise independent, if the random sequence $h(1), h(2), \ldots, h(N)$ is $(k, \epsilon)$-wise independent when $h$ is drawn uniformly at random from $\mathcal{H}$.
\end{definition}

The following statement is implicit in \cite{alon1992simple}, 
and we will provide a proof in the appendix.

\begin{restatable}[\cite{alon1992simple}]{lemma}{almostkwise}\label{lem:hash}
    For any tuple of $(k, \epsilon, N, M)$ where $M$ is an integer power of $2$, there exists a $(k, \epsilon)$-wise independent hash family $\mathcal{H}_{k, \epsilon, N, M}$ mapping from $[N]$ to $[M]$ of size at most $$|\mathcal{H}_{k, \epsilon, N, M}|\leq O\brac{\brac{\frac{1}{\epsilon}}^3M^{3k}\log^3N}$$ %\hongyi{$\log^3N \to \log^3 (N\log M)?$}
\end{restatable}

We need a concentration inequality for almost $k$-wise independent variables. This question was previously studied in \cite{gravin2021concentration} motivated by some applications in cryptography, but we cannot apply their results directly in our scenario because their definitions of almost $k$-wise independence are technically different from ours, 
so we will present a proof of the following statement in the appendix.

\begin{restatable}{lemma}{almostkwiseconc}\label{lem:conc}
    Consider a sequence of random variables $X_1, X_2, \ldots, X_N\in [M]$ which is $(k, \epsilon)$-wise independent for some even integer $k$, along with a sequence of non-negative weights $$w_1, w_2, \ldots, w_N\leq W$$ Define random variables $Y_i = \mathbf{1}[X_i = 1]$. Then for any $\delta$, we have:
    $$\Pr\left[\left|\sum_{i=1}^N w_i\cdot Y_i - \frac{1}{M}\sum_{i=1}^N w_i\right| > \frac{\delta WN}{M}\right]\leq \brac{\frac{M\sqrt{k}}{\delta \sqrt{N}}}^k + \epsilon (2 M^2 / \delta)^k$$
\end{restatable}
\section{The Main Algorithm}
Throughout the algorithm we maintain a partial coloring $\phi$ of the graph using at most $(1+\eta)\Delta$ colors. As a preprocessing step, we apply the algorithm of Lemma~\ref{lemma:prep} iteratively for 
$O\brac{\frac{\alpha}{\eta}\sqrt{\log \Delta}}$ rounds, where $\alpha = 200$ is a constant parameter.
After these rounds, the number of uncolored vertices reduces geometrically, leaving at most 
$n/2^{2\alpha\sqrt{\log \Delta}}$ uncolored vertices. After the preprocessing, the main technical part of the algorithm is summarized as the following statement.
\begin{lemma}
\label{lemma:extend}
Let $\phi$ be a partial coloring with $(1+\eta)\Delta$ colors such that $|\phi^{-1}(\perp)| \leq n / s^2$ for a size parameter $s \ge 2^{\alpha\sqrt{\log \Delta}}$ ($s$ does not need to be an integer). 
There exists a deterministic algorithm that, using $O\brac{\frac{\log\Delta}{\log s}}$ passes and $\tilde{O}(n)$ bits of space, 
extends $\phi$ to a new coloring $\phi'$ satisfying $|\phi'^{-1}(\perp)| < \frac{|\phi^{-1}(\perp)|}{s^{1/10}}$.
\end{lemma}

By iteratively applying Lemma~\ref{lemma:extend}, we can decrease the total number of uncolored vertices from $n / 2^{2\alpha\sqrt{\log \Delta}}$ to $O(n / \Delta)$ within $O(\log\log \Delta)$ iterations. 
Once the number of remaining uncolored vertices is sufficiently small, 
all vertices and their incident edges can be collected into memory and colored using any list coloring algorithm.

\begin{proof}[Proof of~\Cref{theorem:main}] We combine the preprocessing step (\Cref{lemma:prep}), the iterative color extension step (\Cref{lemma:extend}), and a final coloring step to conclude the theorem. 

\paragraph*{Step 1: Preprocessing.} 

As said before, we will repeatedly apply~\Cref{lemma:prep} for $O\brac{\frac{\alpha}{\eta}\sqrt{\log \Delta}}$ rounds reduces the number of uncolored vertices from $n$ to at most $n / 2^{2\alpha\sqrt{\log \Delta}}$. This step uses $O(\sqrt{\log \Delta})$ passes and $\tilde{O}(n)$ bits of space.

\paragraph*{Step 2: Color Extension.} 
Let $\phi_0$ be the partial coloring right after the preprocessing step. The algorithm will iteratively apply $O(\log\log\Delta)$ times \Cref{lemma:extend} to reduce the number of uncolored vertices starting with $\phi_0$.
Let $U_t = \phi_t^{-1}(\perp)$ denote the set of uncolored vertices after the $t$-th iteration, where $\phi_t$ is the partial coloring after the $t$-th iteration. Assume $s_t = \sqrt{n / \left|\phi_t^{-1}(\bot)\right|}$, so initially we have $s_0 \geq 2^{\alpha \sqrt{\log\Delta}}$. Then, in the $(t+1)$-th iteration, apply \Cref{lemma:extend} on the partial coloring $\phi_t$ and size parameter $s_t$, and define $\phi_{t+1} = \phi_t'$ to be the extended coloring output by \Cref{lemma:extend}. We will argue that this iterative process shall terminate after $O(\log\log \Delta)$ iterations and the total number of passes over the data stream is bounded by $O(\sqrt{\log\Delta})$.

\begin{claim}
    After the $t$-th iteration, we have: $s_t \geq 2^{\brac{\frac{21}{20}}^t\cdot \alpha\sqrt{\log\Delta}}$.
\end{claim}
\begin{proof}[Proof of claim]
    This is proved by a direct induction on $t$. As the basis, we already know that $s_0 \geq 2^{\alpha \sqrt{\log\Delta}}$. As for the inductive step, according to the statement of \Cref{lemma:extend}, we have: $$\left|\phi_{t+1}^{-1}(\bot)\right|\leq \left|\phi_t^{-1}(\bot)\right| / s_t^{1/10}\leq n / s_t^{21/10}$$
    By definition of $s_{t+1}$, we have: $s_{t+1}\geq s_{t}^{21/20}\geq 2^{\brac{\frac{21}{20}}^{t+1}\cdot \alpha\sqrt{\log\Delta}}$.
\end{proof}

According to \Cref{lemma:extend}, as the number of passes in the $t$-th iteration is bounded by $O\brac{\frac{\log\Delta}{\log s_t}} = O\brac{\brac{\frac{20}{21}}^t\cdot \sqrt{\log\Delta}}$, the total number of passes throughout all $O(\log\log\Delta)$ iterations is bounded by $O(\sqrt{\log\Delta})$.

\paragraph*{Step 3: In-Memory Coloring.}  
After Step 2, when $t = \Omega(\log\log\Delta)$, only $O(n/\Delta)$ vertices remain uncolored. We devote one additional pass to finish the coloring offline. During this pass, for each vertex $v \in \phi^{-1}(\perp)$ we maintain its current set of available colors. Whenever an edge incident on $v$ appears in the stream, we store it in memory and update the available colors accordingly. Because the remaining subgraph has only $O(n/\Delta)$ vertices, maximum degree at most $\Delta$, and each vertex stores a list of at most $O(\Delta)$ colors, the data structure fits in $O(n\log n)$ space.

After processing all the edges, we have stored in memory the induced subgraph $G[\phi^{-1}(\bot)]$ and all the available colors of each uncolored vertex. Then we can apply the standard greedy $(\deg+1)$-list coloring algorithm to complete the coloring of the remaining vertices.
\end{proof}

\subsection{Algorithm Description}

The rest of this section is devoted to \Cref{lemma:extend}. The algorithm consists of $O\brac{\frac{\log\Delta}{\log s}}$ rounds. Throughout the rounds, the algorithm maintains the following set of information.
\begin{itemize}[leftmargin=*]
    \item \textbf{Vertex Partition.} A partition of the uncolored vertex set $\phi^{-1}(\perp) = B_{\text{bad}}\cup B_1\cup B_2\cdots \cup B_l$, for some parameter $l$. Each vertex set $B_j$ will be called a bin of $\phi^{-1}(\bot)$.
    \item \textbf{Color Partition.} A partition of the color set $[(1+\eta)\Delta] = C_{\text{bad}}\cup C_1\cup C_2\cup \cdots \cup C_l$. Each color set $C_j$ will be called a palette.
    \item \textbf{Frequent Colors.} For each vertex $v\in \phi^{-1}(\bot)$, a set of colors $F_v\subseteq [(1+\eta)\Delta]$. We will make sure that the total size $\bigcup_{v\in \phi^{-1}(\bot)}F_v$ fits in the memory of semi-streaming. Intuitively, the set $F_v$ stores colors that are frequently used by neighbors of $v$ under the current partial coloring $\phi$.
\end{itemize}

At the beginning of the first round, we set $l=1$, $B_{\text{bad}} = C_{\text{bad}} = \emptyset$, $B_1 = \phi^{-1}(\bot)$, and $C_1 = [(1+\eta)\Delta]$. In the $t$-th round, we will divide all the sets among vertex and color partitions further. We will ensure the following invariants at the beginning of the $t$-th round.

\begin{invariant}\label{inv:deg}
    Define parameter $M = \floor{s^{1/200}}_2$. For any $1\leq i\leq l$, we have:
    $$|C_i|\leq (1+\eta)\Delta\cdot \brac{1 + \frac{1}{\log\Delta}}^t / M^t$$
    More importantly, for each $v\in B_i, \forall 1\leq i\leq l$, we have:
    $$|C_i\setminus F_v| - \brac{|\{u\in N(v)\mid u\in B_i\}| + |\{u\in N(v)\mid \phi(u)\in C_i\setminus F_v\}|}\geq \eta\Delta\cdot \brac{1 - \frac{1}{\log\Delta}}^t / M^t$$
    This implies that the number of available colors in $C_i\setminus F_v$ (colors that are not used by neighbors of $v$ under $\phi$) is much larger than the number of neighbors in set $B_i$.
\end{invariant}

Next, we are going to describe how to refine the vertex and color partitions in each round which consists of several steps.

\paragraph*{Pruning Frequent Colors.} The purpose of this procedure is to collect all the colors in $C_i\setminus F_v$ that are currently used frequently by $v$'s neighbors under partial coloring $\phi$, which is done as following. For each vertex $v\in B_i, \forall 1\leq i\leq l$, make a pass over the graph stream and consider all the neighbors $u\in N(v)$ such that $\phi(u)\in C_i\setminus F_v$. Apply \Cref{lem:freq} on the sub-stream of edges $\{\phi(u)\in C_i\setminus F_v\mid u\in N(v)\}$ with space parameter $k = \ceil{s^2}$. This generates a list of colors $L_v = \{c\in C_i\setminus F_v\}$ of at most $s^2$ elements. Since there are at most $n/s^2$ uncolored vertices under $\phi$, we can run all instances of the frequent item algorithm for all $v\in \phi^{-1}(\bot)$ at the same time in a single pass over the input graph stream.After that, merge the list $F_v\leftarrow F_v\cup L_v$.

\paragraph*{Refining Vertex and Color Partitions.} We are going to use almost $k$-wise independent hash functions to partition each bin $B_i$ and palette $C_i$ further. According to \Cref{lem:hash}, by setting the parameters $$k_1 = 10, \epsilon_1 =\epsilon = \brac{\frac{1}{10M^5}}^{k_1}, N_1 = \left|\phi^{-1}(\bot)\right|, M_1 = M= \floor{s^{1/200}}_2$$ 
we can find a $(k_1, \epsilon_1)$-wise independent hash family $\mathcal{H}_1$ of size at most $$\brac{\frac{1}{\epsilon_1}}^3\cdot M_1^{3k_1}\cdot \log^3N_1 = O(s\log^3n)$$ Similarly, by setting the parameters $$k_2 = 10, \epsilon_2 = \epsilon = \brac{\frac{1}{10M^5}}^{k_2}, N_2 = (1+\eta)\Delta, M_2 = M = \floor{s^{1/200}}_2$$
we can find a $(k_2, \epsilon_2)$-wise independent hash family $\mathcal{H}_2$ of size at most $$\brac{\frac{1}{\epsilon_2}}^3\cdot M_2^{3k_2}\cdot \log^3N_2 = O(s\log^3 \Delta)$$

Our goal is to find a good pair of hash functions $(h^*_1, h^*_2)\in \mathcal{H}_1\times \mathcal{H}_2$ and partition each bin $B_j$ and palette $C_j$ as
$$B_i = B_{i, h_1}^1\cup B_{i, h_1}^2\cup \cdots B_{i, h_1}^{M}$$
$$C_i = C_{i, h_2}^1\cup C_{i, h_2}^2\cup \cdots C_{i, h_2}^{M}$$
where $B_{i, h_1}^j = \{v\in B_i\mid h_1^*(v) = j\}, C_{i, h_2}^j = \{c\in C_i\mid h_2^*(c) = j\}$, such that we can verify \Cref{inv:deg} with the refined partitions.

\begin{definition}\label{def:good}
    For any pair of hash functions $(h_1, h_2)\in \mathcal{H}_1\times \mathcal{H}_2$, a vertex $v\in B_i, \forall 1\leq i\leq l$ is \emph{good} (otherwise \emph{bad}) if \Cref{inv:deg} holds for $v$; that is, 
    $$\begin{aligned}
    &\left|C_{i, h_2}^{h_1(v)}\setminus F_v\right| - \brac{\left|\left\{u\in N(v)\mid u\in B_{i, h_1}^{h_1(v)}\right\}\right| + \left|\left\{u\in N(v)\mid \phi(u)\in C_{i, h_2}^{h_1(v)}\setminus F_v\right\}\right|}\\
    &\geq \eta\Delta\cdot \brac{1 - \frac{1}{\log\Delta}}^{t+1} / M^{t+1}
    \end{aligned}$$
    as well as $$\left|C_{i, h_2}^{h_1(v)}\right|\leq (1+\eta)\Delta\brac{1+\frac{1}{\log\Delta}}^{t+1} / M^{t+1}$$
    This pair $(h_1, h_2)$ is \emph{good} if the fraction of bad vertices is less than $s^{-1/9}$.
\end{definition}

The algorithm tries to find a good pair $(h_1, h_2)\in \mathcal{H}_1\times \mathcal{H}_2$ in the straightforward manner. Basically, the algorithm makes one pass over the input graph stream, and each vertex keeps a counter $\cnt(v, h_1, h_2)$ for every pair $(h_1, h_2)\in \mathcal{H}_1\times \mathcal{H}_2$ which is equal to the value of 
$$\cnt(v, h_1, h_2) = \left|C_{i, h_2}^{h_1(v)}\setminus F_v\right| - \brac{\left|\left\{u\in N(v)\mid u\in B_{i, h_1}^{h_1(v)}\right\}\right| + \left|\left\{u\in N(v)\mid \phi(u)\in C_{i, h_2}^{h_1(v)}\setminus F_v\right\}\right|}$$
Maintaining the counter $\cnt(v, h_1, h_2)$ is straightforward: when an edge $(u, v)$ from the input stream is read, we check if $u\in B_{i, h_1}^{h_1(v)}$ and $\phi(u)\in C_{i, h_2}^{h_1(v)}\setminus F_v$ respectively and update $\cnt(v, h_1, h_2)$ accordingly. Besides, we also check whether the second inequality is preserved under $(h_1, h_2)$ for $v$: 
$$\left|C_{i, h_2}^{h_1(v)}\right|\leq (1+\eta)\Delta\brac{1+\frac{1}{\log\Delta}}^{t+1} / M^{t+1}$$
This decides whether a vertex is good or bad with respect to any pair $(h_1, h_2)$, so in the end we can find one good pair $(h_1, h_2)$ if it exists. The total required would be $|\mathcal{H}_1|\cdot |\mathcal{H}_2| = O(s^2\log^6n)$ for each pair, so the overall space would be $\tilde{O}(n)$ since $\left|\phi^{-1}(\bot)\right|\leq n/s^2$.

After we have found a good pair $(h_1^*, h_2^*)\in \mathcal{H}_1\times \mathcal{H}_2$ whose existence will be proven later, add all the bad vertices to $B_{\text{bad}}$, and update the vertex partition as:
$$\phi^{-1}(\bot) = B_{\text{bad}}\cup \bigcup_{i=1}^l\bigcup_{j=1}^M \tilde{B}_{i, h_1^*}^j$$
where $\tilde{B}_{i, h_1^*}^j = \{v\in B_{i, h_1^*}^j\mid v\text{ is good}\}$.

As for the color partition, if some sub-palette $C_{i, h_2^*}^j$ has size larger than $$(1+\eta)\Delta\brac{1+\frac{1}{\log\Delta}}^{t+1} / M^{t+1}$$ then merge it with $C_{\text{bad}}$. All other sub-palettes $C_{i, h_2^*}^j$ stay in the partition:
$$[(1+\eta)\Delta] = C_{\text{bad}}\cup\bigcup_{i=1}^l\bigcup_{1\leq j\leq M, |C_{i, h_2^*}^j|\leq (1+\eta)\Delta\brac{1+\frac{1}{\log\Delta}}^{t+1} / M^{t+1}} C_{i, h_2^*}^j$$

\paragraph*{In-Memory Coloring.} Assume \Cref{inv:deg} is preserved at the beginning of every round. Let $t^*$ be the smallest integer such that $$(1+\eta)\Delta\cdot\brac{1 + \frac{1}{\log\Delta}}^{t^*} / M^{t^*} \leq s^2$$
By definition, $M = \floor{s^{1/200}}_2 > 0.5\cdot 2^{\sqrt{\log\Delta}}$, and so $t^* = O\brac{\frac{\log\Delta}{\log M}} = O\brac{\frac{\log\Delta}{\log s}}$.
Then, after the $(t^*-1)$-th round, in the color partition $[(1+\eta)\Delta] = C_{\text{bad}}\cup C_1\cup \cdots\cup C_l$, for each $1\leq i\leq l$ we have:
$$|C_i| \leq (1+\eta)\Delta\cdot\brac{1 + \frac{1}{\log\Delta}}^{t^*} / M^{t^*} \leq s^2$$
Also, in the vertex partition $\phi^{-1}(\bot) = B_{\text{bad}}\cup B_1\cup \cdots \cup B_l$, for each $1\leq i\leq l$ and $v\in B_i$, we have:
$$\begin{aligned}
    &|C_i\setminus F_v| - \brac{|\{u\in N(v)\mid u\in B_i\}| + |\{u\in N(v)\mid \phi(u)\in C_i\setminus F_v\}|}\\
    &\geq \eta\Delta\cdot \brac{1 - \frac{1}{\log\Delta}}^{t^*} / M^{t^*} > 0
\end{aligned}$$
Consequently, we have:
$$\left|\brac{C_i\setminus F_v}\setminus \{u\in N(v)\mid \phi(u)\in C_i\setminus F_v\} \right| > |\{u\in N(v)\mid u\in B_i\}|$$
This means that the total number of neighbors $u\in N(v)$ that are in the same bin $B_i$ as $v$ is strictly less than the number of available colors of $v$ in $\brac{C_i\setminus F_v}\setminus \{u\in N(v)\mid \phi(u)\in C_i\setminus F_v\}$. So, there exists a color extension $\phi'$ of $\phi$ which assigns a color from $C_i\setminus F_v$ to $v$ for each $v\in B_i, \forall 1\leq i\leq l$. To find such a valid color extension $\phi'$, it suffices to store all the color lists $C_i\setminus F_v$ as well as all edges $(u, v)\in E[B_i], \forall 1\leq i\leq l$ in memory. Since the sets $\brac{C_i\setminus F_v}\setminus \{u\in N(v)\mid \phi(u)\in C_i\setminus F_v\}$ and $\{u\in N(v)\mid u\in B_i\}$ all have sizes at most $s^2$, we can compute and store them in memory in one pass over the data stream. The whole algorithm is summarized in \Cref{alg:color-extend}.

\iffalse
\begin{algorithm}
    \caption{$\mathsf{ColorExtend}(\phi)$}\label{alg:color-extend}
    $s\leftarrow \sqrt{n / \left|\phi^{-1}(\bot)\right|}$\;
    $M\leftarrow \floor{s^{1/200}}_2$\;
    start with trivial vertex and color partitions $B_{\text{bad}}, C_{\text{bad}}\leftarrow \emptyset, B_1 \leftarrow \phi^{-1}(\bot), C_1\leftarrow [(1+\eta)\Delta]$\;
    \For{$t = 0, 1, \ldots, O\brac{\frac{\log\Delta}{\log s}}$}{
        build almost $10$-wise independent hash families $\mathcal{H}_1, \mathcal{H}_2$\;
        find a good pair $(h_1^*, h_2^*)\in \mathcal{H}_1\times \mathcal{H}_2$ in a single pass over the input stream\;
        add all bad vertices to $B_{\text{bad}}$ and update the vertex partition as: $\phi^{-1}(\bot) = B_{\text{bad}}\cup \bigcup_{i=1}^l\bigcup_{j=1}^M \tilde{B}_{i, h_1^*}^j$\;
        also, merge the large sub-palettes $C_{i, h_2^*}^j$ with $C_{\text{bad}}$, and update the color partition as:
        $[(1+\eta)\Delta] = C_{\text{bad}}\cup\bigcup_{i=1}^l\bigcup_{j\text{ under some condition}} C_{i, h_2^*}^j$\;
    }
    compute and store all the color sets $\brac{C_i\setminus F_v}\setminus \{u\in N(v)\mid \phi(u)\in C_i\setminus F_v\}$ and subgraphs $G[B_i]$ in memory, $\forall 1\leq i\leq l$\;
    compute a list coloring which assigns colors to all vertices in $\phi^{-1}(\bot)\setminus B_{\text{bad}}$\;
\end{algorithm}
\fi
\begin{algorithm}[H]
    \caption{$\mathsf{ColorExtend}(\phi)$}\label{alg:color-extend}
    \SetKwInOut{Input}{Input}
    \Input{Partial coloring $\phi$ with at most $(1+\eta)\Delta$ colors}
    
    $s \leftarrow \sqrt{n / |\phi^{-1}(\perp)|}$\;
    $M \leftarrow \lfloor s^{1/200} \rfloor_2$\;
    Initialize $F_v \leftarrow \emptyset$ for all $v \in \phi^{-1}(\perp)$\;
    Start with trivial partitions: $B_{\text{bad}} \leftarrow \emptyset, C_{\text{bad}} \leftarrow \emptyset, B_1 \leftarrow \phi^{-1}(\perp), C_1 \leftarrow [(1+\eta)\Delta]$\;
    
    \For{$t = 0, 1, \ldots, O\left(\frac{\log\Delta}{\log s}\right)$}{
        \tcp{Step 1: Pruning Frequent Colors}
        Make one pass over the stream to find frequent colors $L_v \subseteq C_i \setminus F_v$ using \Cref{lem:freq} with $k = \lceil s^2 \rceil$ for all $v \in \phi^{-1}(\perp)$ simultaneously\;
        Update $F_v \leftarrow F_v \cup L_v$ for each $v \in \phi^{-1}(\perp)$\;
        
        \tcp{Step 2: Refining Partitions}
        Build $(10, \epsilon)$-wise independent hash families $\mathcal{H}_1, \mathcal{H}_2$ with parameters from text\;
        Find a good pair $(h_1^*, h_2^*) \in \mathcal{H}_1 \times \mathcal{H}_2$ in a single pass according to \Cref{def:good}\;
        
        \tcp{Step 3: Update Vertex and Color Bins}
        Add all bad vertices (that violate Invariant \ref{inv:deg}) to $B_{\text{bad}}$\;
        Update vertex partition: $\phi^{-1}(\perp) = B_{\text{bad}} \cup \bigcup_{i, j} \tilde{B}_{i, h_1^*}^j$\;
        Move large sub-palettes $|C_{i, h_2^*}^j| > (1+\eta)\Delta \frac{(1+1/\log\Delta)^{t+1}}{M^{t+1}}$ to $C_{\text{bad}}$\;
        Update color partition: $[(1+\eta)\Delta] = C_{\text{bad}} \cup \bigcup_{i, j} C_{i, h_2^*}^j$\;
    }
    
    \tcp{Final Step: In-Memory Coloring}
    Compute and store color sets $(C_i \setminus F_v) \setminus \{ \phi(u) : u \in N(v), \phi(u) \in C_i \setminus F_v \}$ and induced subgraphs $G[B_i]$ for all $1 \le i \le l$ in one pass\;
    Apply greedy list coloring to find $\phi'$ for all $v \in \phi^{-1}(\perp) \setminus B_{\text{bad}}$\;
    \Return $\phi'$\;
\end{algorithm}

\subsection{Proof of Correctness}
We first show that by merging the list $L_v$ with $F_v$, we have removed all the frequent colors from consideration.
\begin{lemma}\label{lem:remove-freq}
    Let $v\in B_i, i\geq 1$ be any vertex. After the algorithm merges the lists $L_v$ with sets $F_v$, for any color $c\in C_i\setminus F_v$, we have:
    $$|\{u\in N(v)\mid \phi(u) = c\}|\leq \frac{|C_i\setminus F_v|}{s}+2$$
    At the same time, we also have:
    $$\begin{aligned}
        |C_i\setminus F_v| - \brac{|\{u\in N(v)\mid u\in B_{i}\}| + |\{u\in N(v)\mid \phi(u)\in C_{i}\setminus F_v\}|}
        \geq \frac{\eta}{3}\cdot s^2 - \ceil{s}
    \end{aligned}$$
\end{lemma}
\begin{proof}
    According to \Cref{lem:freq}, $L_v$ contains at most $\ceil{s}$ elements but it includes all elements in $C_i\setminus F_v$ (before merging $L_v$ with $F_v$) such that $|\{u\in N(v)\mid \phi(u) = c\}| > \frac{|C_i\setminus F_v|}{s}$. Therefore, after merging $F_v$ with $L_v$, we have:
    $$|\{u\in N(v)\mid \phi(u) = c\}| \leq \frac{|C_i\setminus F_v| + \ceil{s}}{s} \leq \frac{|C_i\setminus F_v|}{s}+2$$
    Also, by \Cref{inv:deg} and that $|L_v|\leq \ceil{s}$, we have:
    $$\begin{aligned}
        &|C_i\setminus F_v| - \brac{|\{u\in N(v)\mid u\in B_{i}\}| + |\{u\in N(v)\mid \phi(u)\in C_{i}\setminus F_v\}|}\\
        &\geq \eta\Delta\brac{1 - \frac{1}{\log\Delta}}^t/M^t - \ceil{s}\\
        &= \frac{\eta}{1+\eta}\brac{\frac{\log\Delta-1}{\log\Delta+1}}^t\cdot (1+\eta)\Delta\brac{1+\frac{1}{\log\Delta}}^t / M^t - \ceil{s}\\
        &\geq \frac{\eta}{3}\cdot s^2 - \ceil{s}
    \end{aligned}$$
    The last inequality has utilized the fact that $\frac{1}{1+\eta}\brac{\frac{\log\Delta-1}{\log\Delta+1}}^t\geq \frac{1}{(1+\eta)e} > \frac{\eta}{3}$ when $\eta < 0.1$, as well as the fact that $(1+\eta)\Delta\brac{1+\frac{1}{\log\Delta}}^t / M^t > s^2$ when $t<t^*$, according to the definition of $t^*$ in the previous sub-section.
\end{proof}

Next, we show that there always exists a good pair $(h_1^*, h_2^*)$ of hash functions.
\begin{lemma}
    There always exists a good pair $(h_1^*, h_2^*)\in \mathcal{H}_1\times \mathcal{H}_2$ of hash functions under \Cref{def:good}.
\end{lemma}
\begin{proof}
    Let us analyze the probability that any single vertex $v\in B_i, \forall 1\leq i\leq l$ is good when $h_1, h_2$ are drawn uniformly at random from $\mathcal{H}_1, \mathcal{H}_2$. To lower bound the value of linear estimator which is $$\psi(v, j, h_1, h_2) \overset{\text{def}}{=} \left|C_{i, h_2}^{j}\setminus F_v\right| - \left|\left\{u\in N(v)\mid u\in B_{i, h_1}^j\right\}\right| - \left|\left\{u\in N(v)\mid \phi(u)\in C_{i, h_2}^{j}\setminus F_v\right\}\right| $$
    We view it as a combination of three sums:
    $$\begin{aligned}
        \psi(v, j, h_1, h_2) &= \left|C_{i, h_2}^{j}\setminus F_v\right| - \left|\left\{u\in N(v)\mid u\in B_{i, h_1}^j\right\}\right| - \left|\left\{u\in N(v)\mid \phi(u)\in C_{i, h_2}^{j}\setminus F_v\right\}\right| \\
        &= \sum_{c\in C_i\setminus F_v}\indic{h_2(c) = j} - \sum_{u\in N(v)\cap B_i}\indic{h_1(u) = j}\\
        &- \sum_{c\in C_i\setminus F_v}\indic{h_2(c) = j}\cdot \left|\{u\in N(v)\mid \phi(u) = c\}\right| 
    \end{aligned}$$
    Then we are going to apply \Cref{lem:conc} to each of the sums respectively.
    \begin{itemize}[leftmargin=*]
        \item For the first summation, by \Cref{lem:remove-freq}, we know that:
        $$|C_i\setminus F_v| \geq \frac{\eta s^2}{3} -\ceil{s} > s^{1/2}$$
        Then, by \Cref{lem:conc}, for parameter $\gamma = \frac{\eta}{20\log\Delta}, k=10$ we have:
        \begin{equation}\begin{aligned}\label{equ:sum1}
            &\Pr_{h_2\sim \mathcal{H}_2}\left[\left|\sum_{c\in C_i\setminus F_v}\indic{h_2(c)=j} - \frac{|C_i\setminus F_v|}{M}\right| \geq \frac{\gamma |C_i\setminus F_v|}{M}\right]\\
            &\leq \brac{\frac{M\sqrt{k}}{\gamma\sqrt{|C_i\setminus F_v|}}}^k + \epsilon (2 M^2 / \gamma)^k\leq \brac{\frac{20M\sqrt{k}\log\Delta}{\eta s^{1/4}}}^k + \brac{\frac{40\log\Delta}{\eta M^3}}^k
        \end{aligned}
        \end{equation}
        By definition $k=10, M = \floor{s^{1/200}}_2$ and $s\geq 2^{200\sqrt{\log\Delta}}$, we have:
        $$\frac{20M\sqrt{k}\log\Delta}{\eta s^{1/4}} \leq \frac{20\sqrt{10}\cdot s^{1/200}\log\Delta}{\eta s^{1/4}} < \frac{1}{s^{1/5}}$$
        $$\frac{40\log\Delta}{\eta M^3} \leq \frac{640\log\Delta}{\eta s^{3/200}} < \frac{1}{s^{1/80}}$$
        Therefore, the right-hand side of \Cref{equ:sum1} is less than:
        $$\brac{\frac{1}{s^{1/5}}}^k + \brac{\frac{1}{s^{1/80}}}^k < 2s^{-1/8}$$

        \item For the second summation, if $|N(u)\cap B_i| \leq s^{1/2}$, then by \Cref{lem:remove-freq} we know that
        $$|N(u)\cap B_i| \leq s^{1/2} < \frac{\eta s}{3} -2\leq |C_i\setminus F_v| /s$$ 
        Otherwise if $|N(u)\cap B_i| > s^{1/2}$, then by \Cref{lem:conc}, for parameter $\gamma = \frac{\eta}{20\log\Delta}$ we have (reusing the same calculation as above):
        \begin{equation}\label{equ:sum2}
        \begin{aligned}
            &\Pr_{h_1\sim \mathcal{H}_1}\left[\left|\sum_{u\in N(v)\cap B_i}\indic{h_1(u)=j} - \frac{|N(v)\cap B_i|}{M}\right| \geq \frac{\gamma |N(v)\cap B_i|}{M}\right]\\
            &\leq \brac{\frac{M\sqrt{k}}{\gamma\sqrt{|N(v)\cap B_i|}}}^k + \epsilon(2 M^2 / \gamma)^k \leq \brac{\frac{20M\sqrt{k}\log\Delta}{\eta s^{1/4}}}^k + \brac{\frac{40\log\Delta}{\eta M^3}}^k \leq 2s^{-1/8}
        \end{aligned}
        \end{equation}

        \item Finally, let us turn to the third summation. Classify all the colors $c\in C_i\setminus F_v$ currently used by some neighbors of $v$ as following. For each non-negative integer $r\geq 0$, define a color subset $$D_r = \left\{c\in C_i\setminus F_v\mid |\{u\in N(v)\mid \phi(u) = c\}|\in [2^r, 2^{r+1})\right\}$$
        According to \Cref{lem:remove-freq}, when $2^r > \frac{|C_i\setminus F_v|}{s}+2$ it must be $D_r = \emptyset$. Additionally, define the following quantity:
        $$d_r = \sum_{c\in D_r}\left|\{u\in N(v)\mid \phi(u)\in D_r\}\right|$$
        So by definition, we have:
        $$\sum_{r\geq 0}d_r = \left|\left\{u\in N(v)\mid \phi(u)\in C_i\setminus F_v\right\}\right|$$
        For each index $r$ such that $|D_r| > \sqrt{s}$, apply \Cref{lem:conc} with parameter $\gamma = \frac{1}{5\log\Delta}$ we have (reusing the same calculation as above):
        \begin{equation}\label{equ:sum3}
            \begin{aligned}
            &\Pr_{h_2\sim \mathcal{H}_2}\left[\left|\sum_{c\in D_r}\indic{h_2(c)=j}\cdot |\{u\in N(v)\mid \phi(u) = c\}| - \frac{d_r}{M}\right| \geq \frac{\gamma/2 \cdot 2^{r+1}|D_r|}{M} > \frac{\gamma d_r}{M}\right]\\
            &\leq \brac{\frac{M\sqrt{k}}{\gamma\sqrt{|D_r|}}}^k + \epsilon(2 M^2 / \gamma)^k\leq \brac{\frac{20M\sqrt{k}\log\Delta}{\eta s^{1/4}}}^k + \brac{\frac{40\log\Delta}{\eta M^3}}^k \leq 2s^{-1/8}
            \end{aligned}
        \end{equation}
        For the rest cases where $|D_r| \leq \sqrt{s}$, we have:
        \begin{equation}\label{equ:sum4}
            \sum_{|D_r|\leq \sqrt{s}} d_r \leq \log\Delta\cdot \sqrt{s}\cdot 2\brac{\frac{|C_i\setminus F_v|}{s}+2} < 3\log\Delta \cdot |C_i\setminus F_v| / \sqrt{s}
        \end{equation}
    \end{itemize}

    Summing up the above three cases, we can conclude (by a union bound) that with probability at least $1 - 6M s^{-1/8}$ over the choice of $(h_1, h_2)$, the linear estimator $\psi(v, j, h_1, h_2)$ is at least (for all $j\in [M]$ ):
    $$\begin{aligned}
        \psi(v, j, h_1, h_2) &\geq \frac{(1-\gamma)|C_i\setminus F_v|}{M} - \brac{\frac{(1+\gamma)|N(v)\cap B_i|}{M} + \frac{|C_i\setminus F_v|}{s}}\\
        &- \brac{\frac{(1+\gamma)\left|\left\{u\in N(v)\mid \phi(u)\in C_i\setminus F_v\right\}\right|}{M} + \frac{3\log\Delta\cdot |C_i\setminus F_v|}{\sqrt{s}}}\\
        &\geq \frac{1}{M}\cdot\brac{|C_i\setminus F_v| - |\{u\in N(v)\mid u\in B_i\}| - |\{u\in N(v)\mid \phi(u)\in C_i\setminus F_v\}|}\\
        &- \frac{\gamma}{M}\cdot \brac{|C_i\setminus F_v| +|\{u\in N(v)\mid u\in B_i\}| + |\{u\in N(v)\mid \phi(u)\in C_i\setminus F_v\}|}\\
        &- \brac{\frac{|C_i\setminus F_v|}{s} + \frac{3\log\Delta\cdot |C_i\setminus F_v|}{\sqrt{s}}}\\
        &\geq \eta\Delta\brac{1 - \frac{1}{\log\Delta}}^t / M^{t+1} - 2\gamma\cdot (1+\eta)\Delta\brac{1 + \frac{1}{\log\Delta}}^t / M^{t+1} \\        
        &-\brac{\frac{M}{s} + \frac{3M\log\Delta}{\sqrt{s}}}\cdot (1+\eta)\Delta\brac{1 + \frac{1}{\log\Delta}}^t / M^{t+1}\\
        &\geq \brac{1 - \frac{(1+\eta)e^2}{10\log\Delta} - \frac{(1+\eta)e^2}{\eta}\brac{s^{-0.9} + s^{-0.4}}}\cdot \eta\Delta\brac{1 - \frac{1}{\log\Delta}}^t / M^{t+1} \\
        &\geq \eta\Delta\brac{1 -\frac{1}{\log\Delta}}^{t+1} / M^{t+1}
    \end{aligned}$$
    Here the first and second inequality is by \Cref{equ:sum1,equ:sum2,equ:sum3,equ:sum4} and union bound, and the third inequality is by \Cref{inv:deg} which holds at the beginning of the $t$-th round.
    
    Finally, let us also upper bound the probability that the second condition of \Cref{def:good} is violated. Again, for any $j\in [M]$, using \Cref{lem:conc} we have:
    $$\begin{aligned}
    	&\Pr_{h_2\sim \mathcal{H}_2}\left[ \left|C_{i, h_2}^{j}\right|> (1+\eta)\Delta\brac{1+\frac{1}{\log\Delta}}^{t+1} / M^{t+1}\right] \\
    	&= \Pr_{h_2\sim \mathcal{H}_2}\left[ \sum_{c\in C_i}\indic{h_2(c) = j}> (1+\eta)\Delta\brac{1+\frac{1}{\log\Delta}}^{t+1} / M^{t+1}\right]\\
    	&\leq \Pr_{h_2\sim \mathcal{H}_2}\left[\left| \sum_{c\in C_i}\indic{h_2(c) = j} - \frac{|C_i|}{M}\right|> \frac{|C_i|}{M\log\Delta}\right]\\
    	&\leq \brac{\frac{M\sqrt{k}\log\Delta}{\sqrt{|C_i|}}}^k + \epsilon(2M^2\log\Delta)^k \leq 2s^{-1/8}
    \end{aligned}$$
    Here the first inequality has used the fact that \Cref{inv:deg} holds in the previous round, and the last inequality holds because it is even smaller compared to right-hand side of \Cref{equ:sum1}.
    
    By a union bound over all $j\in [M]$, we can conclude that $v$ is good under \Cref{def:good} with probability at least $1 - 8M s^{-1/8} > 1 - s^{-1/9}$ when $(h_1, h_2)$ is drawn uniformly at random from $\mathcal{H}_1\times\mathcal{H}_2$. By linearity of expectation and according to \Cref{def:good}, there must be a good pair $(h_1^*, h_2^*)\in \mathcal{H}_1\times\mathcal{H}_2$ which concludes the proof.    
\end{proof}

Now that we have shown that a good pair $(h_1^*, h_2^*)$ always exists, the algorithm would use it to partition the color and vertex sets further while still preserving \Cref{inv:deg}. By the algorithm description, the number of rounds is at most $O\brac{\frac{\log\Delta}{\log s}}$, so the total number of vertices in $B_{\text{bad}}$ is at most $$|B_{\text{bad}}|\leq O\brac{\frac{\log\Delta}{\log s}}\cdot s^{-1/9}\cdot \left|\phi^{-1}(\bot)\right| \leq s^{-1/10}\cdot \left|\phi^{-1}(\bot)\right|$$
using the fact that $s\geq 2^{\alpha\sqrt{\log\Delta}}$ with $\alpha=200$ being a large constant. By the algorithm description, all vertices in $\phi^{-1}(\bot)\setminus B_{\text{bad}}$ would be colored under the color extension $\phi'$, so this concludes the proof of \Cref{lemma:extend}.
\section{Open Problems}
Our result for streaming vertex coloring initiates several interesting questions regarding deterministic streaming algorithms.

\paragraph*{Optimal Pass Complexity.} The current number of passes $O(\sqrt{\log\Delta})$ for $O(\Delta)$-coloring seems to be a barrier against our approach. New ideas might be needed to achieve faster pass efficiency, say $O(\log\log\Delta)$ passes.

\paragraph*{$(\Delta+1)$-Coloring.} Bypassing the logarithmic bound $O(\log\Delta)$ in the case of $(\Delta+1)$-coloring looks significantly more challenging than $O(\Delta)$-coloring. Both Step $1$ and Step $2$ break down in the proof of \Cref{theorem:main}: firstly, it is not known how to reduce the number of uncolored vertices by half within $O(1)$ passes; secondly, the approach of pruning frequent neighbor colors also does not work with only $\Delta+1$ colors.

\paragraph*{Maximal Independent Set.} As a closely related problem to $(\Delta+1)$-coloring, maximal independent set is also a locally checkable problem, but it is usually harder than $(\Delta+1)$-coloring in various computation models (for example, in the dynamic setting, $(\Delta+1)$-coloring seems easier than maximal independent set \cite{10.1145/3501403,DBLP:journals/talg/BhattacharyaGKL22,DBLP:conf/focs/BehnezhadDHSS19,DBLP:conf/focs/ChechikZ19}), and arguably studied more widely than $(\Delta+1)$-coloring. In the streaming setting, the pass complexity of randomized algorithms has been settled recently in \cite{10.1145/3618260.3649763} which is $\Theta(\log\log n)$, while the current upper bound for deterministic algorithms is $O(\log n)$ by derandomizing Luby's algorithm \cite{Luby1985ASP}. Breaking the logarithmic bound for deterministic maximal independent set would raise a fundamental question.

\section*{Acknowledgment}
Shiri Chechik is supported by the European Research Council (ERC)
under the European Union’s Horizon 2020 Research and Innovation program, grant agreement
No. 803118 ``The Power of Randomization in Uncertain Environments''. Hongyi Chen and Tianyi Zhang are supported by the Fundamental and Interdisciplinary Disciplines Breakthrough Plan of the Ministry of Education of China (No. JYB2025XDXM118) and the ``111 Center'' (No. B26023).

\vspace{5mm}
\bibliographystyle{alpha}
\bibliography{ref-dblp}

\appendix

\section{Proof of \Cref{lemma:prep}}

\extendcolor*
\iffalse
\begin{lemma}[restate \Cref{lemma:prep}]
Let $G$ be an $n$-vertex graph with maximum degree $\Delta$, presented in an insertion-only stream, and let $\phi$ be a partial coloring using at most $(1+\eta)\Delta$ colors.  There exists a deterministic algorithm that, using $O(1)$ passes and $\tilde{O}(n)$ bits of space, extends $\phi$ to a new coloring $\phi'$ such that 
$|\phi'^{-1}(\perp)| < \brac{1 - \Omega(\eta)} |\phi^{-1}(\perp)|$; in other words, the number of uncolored vertices drops by a $1 - \Omega(\eta)$ factor.
\end{lemma}
\fi
%\hongyi{The space is $\Theta(n)$ instead of $\Theta(n / \eta)?$}
\begin{proof}
    For a comparison, in \cite{assadi2022deterministic}, the authors used near-universal hash functions and the algorithm runs in two passes over the data stream, whereas our algorithm will use universal hash functions \cite{carter1977universal} and takes $O(1)$ passes. By the prime number theorem \cite{Grosswald1984}, there exists a prime $p\in \left[(1+\eta/2)\Delta, (1+\eta)\Delta\right]$ when $\eta$ is a constant and $\Delta$ is non-constant.
    % \hongyi{Bertrand’s postulate only guarantees there exists $p \in \left[(1+\eta/2)\Delta, 2(1+\eta/2)\Delta\right]$, why such $p$ can always be found? And it seems that $(1+\eta)\Delta$ is not used in the later proofs.} \tianyi{Use prime number theorem \cite{Grosswald1984}} 
    Define $d = \ceil{\log_p n}$, and consider any $d$-dimensional vector $\overrightarrow{x}\in [p]^d$ and $b\in [p]$. For any vertex $v\in V$, as $p^d \geq n$, we can encode it as a unique vector $\overrightarrow{v}\in [p]^d$. Consider the hash function $$h_{\overrightarrow{x}, b}(v) = \overrightarrow{x}\cdot \overrightarrow{v} + b\mod p$$
    For any hash function $h_{\overrightarrow{x}, b}$, as in \cite{assadi2022deterministic}, edge $(u, v)\in E$ is called monochromatic (under the choice of $h_{\overrightarrow{x}, b}$) if (1) $\phi(u) = \phi(v)= \bot$ and $h_{\overrightarrow{x}, b}(u) = h_{\overrightarrow{x}, b}(v)$, or (2) $\phi(v) = \bot$ and $h_{\overrightarrow{x}, b}(v) = \phi(u)$. Clearly, when $\overrightarrow{x}, b$ are chosen uniformly at random, the probability that any edge $(u, v)$ being monochromatic is exactly $1/p$. Therefore, by linearity of expectation, there exists a pair $(\overrightarrow{x}_*, b_*)$ whose number of monochromatic edges is at most $\frac{\Delta\cdot |\phi^{-1}(\bot)|}{p}\leq \frac{2}{2+\eta}\cdot |\phi^{-1}(\bot)|$. Define an extension $\phi'$ of $\phi$ as:
    $$\phi'(v) = \begin{cases}
        \phi(v) &   \phi(v)\neq \bot\\
        h_{\overrightarrow{x}, b}(v)    &   \phi(v) = \bot \text{ and $v$ is not incident on any monochromatic edges}\\
        \bot    &   \text{else}
    \end{cases}$$
    Then $\phi'$ is still a valid partial coloring and $|\phi'^{-1}(\perp)| < \frac{2}{2+\eta}\cdot |\phi^{-1}(\perp)|$.

    To find the pair $(\overrightarrow{x}_*, b_*)$ in $O(1)$ passes over the data stream, we are using the standard method of conditional expectation as in \cite{assadi2023coloring,czumaj2020simple}. That is, rewriting each possible $(\overrightarrow{x}, b)$ as a vector $(y_1, y_2, \ldots, y_N)\in [p]^N, N = \ceil{\log_p n}+1$, in each pass over the data stream, we are going to fix $\floor{\log_p n}$ more entries in the vector $(y_1, y_2, \ldots, y_N)$. So the total number passes required is $O(1)$. 
    % \hongyi{Maybe $p > n$?} \tianyi{No, $p$ is less than palette size}

    Assume we have already fixed a prefix $y_1^*, y_2^*, \ldots, y_l^*$, such that the expected number of monochromatic edges of hash function $h_{y_1^*, y_2^*, \ldots, y_{l}^*, y_{l+1}, \ldots, y_N}$ when $y_{l+1}, \ldots, y_N$ are chosen uniformly at random from $[p]$, is bounded by $\frac{\Delta\cdot |\phi^{-1}(\bot)|}{p}$. To fix the next (at most) $\floor{\log_p n}$ entries $y_{l+1}^*, y_{l+2}^*, \ldots, y_{\max\{l+\floor{\log_p n}^*, N\}}$, for each possible choice $\overrightarrow{z}$ for $y_{l+1}^*, y_{l+2}^*, \ldots, y_{\max\{l+\floor{\log_p n}^*, N\}}$, the algorithm computes the expected number $M_{\overrightarrow{z}}$ of monochromatic edges of $$h_{y_1^*, y_2^*, \ldots, y_l^*, \overrightarrow{z}, y_{l+1+\floor{\log_p n}}, \ldots, y_N}$$ where $y_{l+1+\floor{\log_p n}}, \ldots, y_N$ is chosen uniformly at random from $[p]$. All values $M_{\overrightarrow{z}}$ can be computed in one pass using $\tilde{O}(n)$ bits of memory, by maintaining a sum $$\sum_{(u, v)\in E}\Pr\left[(u, v)\text{ is monochromatic under }h_{y_1^*, y_2^*, \ldots, y_l^*, \overrightarrow{z}, y_{l+1+\floor{\log_p n}}, \ldots, y_N}\right]$$ for each different $\overrightarrow{z}$.
\end{proof}

\section{Proof of \Cref{lem:hash}}
For almost $k$-wise independence, the original paper \cite{alon1992simple} only studied hash function constructions taking binary values. Their main result can be summarized as follows.
\begin{lemma}[original statement from \cite{alon1992simple}]\label{lem:alon}
    For any tuple of $(k, \epsilon, N)$, there exists a $(k, \epsilon)$-wise independent hash family $\mathcal{H}_{k, \epsilon, N}$ whose domain and range are $[N]$ and $\{0, 1\}$, and every function in $\mathcal{H}_{k, \epsilon, N}$ can be described using at most $(2+o(1))\cdot\brac{\log\log N + k/2+ \log k + \log\frac{1}{\epsilon}}$ bits.
\end{lemma}

\iffalse
\begin{lemma}[restate \cite{alon1992simple}]
    For any tuple of $(k, \epsilon, N, M)$ where $M$ is an integer power of $2$ and $k\geq 4$, there exists a $(k, \epsilon)$-wise independent hash family $\mathcal{H}_{k, \epsilon, N, M}$ of size at most $|\mathcal{H}_{k, \epsilon, N, M}|\leq O\brac{\brac{\frac{1}{\epsilon}}^3M^{3k}\log^3N}$.
\end{lemma}
\fi

\almostkwise*
\begin{proof}
    Since $\log M$ is an integer, we can apply \Cref{lem:alon} with parameter $N' = N\log M, k' = k\log M$, we obtain a $(k', \epsilon)$-wise independent hash family $\mathcal{H}'_{k', \epsilon, N'}$ with binary range and size at most $\brac{\frac{1}{\epsilon}}^3M^{3k}\log^3N'$. For each hash function $f\in \mathcal{H}'_{k', \epsilon, N'}$, concatenate the outputs of every consecutive $\log M$ inputs, it gives a binary string of length $\log M$ which corresponds to a value in $[M]$, and so it can be viewed as a function mapping from $[N]$ to $[M]$.
\end{proof}

\section{Proof of \Cref{lem:conc}}
The proof requires a moment estimation for (ordinary) $k$-wise independent variables. 
\begin{lemma}[\cite{skorski2022tight}]\label{lem:kwise-conc}
    Consider a sequence of random variables $X_1, X_2, \ldots, X_N$ satisfying the following conditions.
    \begin{itemize}
        \item $\{X_i\}_{1\leq i\leq N}$ are $k$-wise independent and $|X_i - \mathbb{E}[X_i]|\leq 1$.
        \item $\sum_{i=1}^N \mathsf{Var}[X_i]\leq N\sigma^2$.
    \end{itemize}
    Then, for $S = \sum_{i=1}^N X_i$ and any positive even integer $d\leq k$ we have:
    $$\mathbb{E}\left[\brac{S - \mathbb{E}[S]}^d\right]\leq 
    \begin{cases}
        (dN)^{d/2}\sigma^d  &   \log(d/N\sigma^2) < \max\{d/N, 2\}\\
        \brac{\frac{d}{\log(d/N\sigma^2)}}^d    &    \max\{d/N, 2\} \leq \log(d/N\sigma^2) \leq d\\
        N\sigma^2   &   d<\log(d/N\sigma^2)
    \end{cases}$$
\end{lemma}

\iffalse
\begin{lemma}[restate \Cref{lem:conc}]
    Consider a sequence of random variables $X_1, X_2, \ldots, X_N\in [M]$ which is $(k, \epsilon)$-wise independent for some even integer $k$, along with a sequence of non-negative weights $w_1, w_2, \ldots, w_N\leq W$. Define random variables $Y_i = \mathbf{1}[X_i = 1]$. Then for any $\delta$, we have:
    $$\Pr\left[\left|\sum_{i=1}^N w_i\cdot Y_i - \frac{1}{M}\sum_{i=1}^N w_i\right| > \frac{\delta WN}{M}\right]\leq \brac{\frac{M\sqrt{k}}{\delta \sqrt{N}}}^k + \epsilon(2 M^2/ \delta)^k$$
\end{lemma}
\fi
\almostkwiseconc*
\begin{proof}
    By Markov's inequality, we have:
    \begin{equation}\label{equ:markov}
    \begin{aligned}
        \Pr\left[\left|\sum_{i=1}^N w_i\cdot Y_i - \frac{1}{M}\sum_{i=1}^N w_i\right| > \frac{\delta WN}{M}\right] &= \Pr\left[\brac{\sum_{i=1}^N w_i\cdot Y_i - \frac{1}{M}\sum_{i=1}^N w_i}^k > \frac{(\delta WN)^k}{M^k}\right]\\
        &\leq \frac{\mathbb{E}\left[\brac{\sum_{i=1}^N \frac{w_i}{W}\cdot(Y_i - 1/M)}^k\right]}{\frac{(\delta N)^k}{M^k}}
    \end{aligned}
    \end{equation}
    Here we have used the fact that $k$ is an even integer. For notational convenience, define $Z_i = \frac{w_i}{W}\cdot (Y_i - 1/M)$. To estimate the enumerator, expand the $k$-th power $\brac{\sum_{i=1}^N Z_i}^k$ as a summation of terms $\Pi_{j=1}^l Z_{i_j}^{n_j}$ where $\sum_{j=1}^l n_j = k, n_j\geq 1$. Next, we study the expectation of each such term $\Pi_{j=1}^k Z_{i_j}^{n_j}$. Let us first prove the following claim.
    % \begin{claim}
    %     For any fixed binary vector $b_1, b_2, \ldots, b_l\in \{0, 1\}$, define $\phi(b) = \begin{cases}
    %         1-\frac{1}{M}   &   b = 0\\
    %         \frac{1}{M}     &   b = 1
    %     \end{cases}$. Then, we have:
    %     $$\left|\Pr\left[Y_{i_j} = b_j, \forall 1\leq j\leq l\right] - \prod_{j=1}^l \phi(b_j)\right| \leq \epsilon\cdot M^k$$
    % \end{claim}
    \begin{claim}
    Let $\phi : \{0,1\} \to \mathbb{R}$ be defined as $\phi(b) = \begin{cases}
                1-\frac{1}{M}   &   b = 0\\
                \frac{1}{M}     &   b = 1
    \end{cases}$.
    For any $l \le k$ and any distinct indices $i_1, \dots, i_l$, and any fixed vector $(b_1, \dots, b_l) \in \{0,1\}^l$, we have
    $$\left|\Pr\left[Y_{i_j} = b_j, \forall 1\leq j\leq l\right] - \prod_{j=1}^l \phi(b_j)\right| \leq \epsilon\cdot M^k.$$
    \end{claim}
        
    \begin{proof}[Proof of claim]
        Using triangle inequality and by the definition of $(k, \epsilon)$-wise independence, for any integer sequence $c_1, c_2, \ldots, c_l\in [M]$, we have:
        \begin{equation}\label{equ:almost-k-wise}
            \left|\Pr\left[X_{i_j} = c_j, \forall 1\leq j\leq l\right] - 1/M^l\right|\leq \epsilon M^{k-l}
        \end{equation}
        Define a set of tuples 
        $T = \left\{(c_1, c_2, \ldots, c_l)\in [M]^l\mid c_j = 1 \text{ iff } b_j = 1\right\}$. so, $|T| = M^l\prod_{j=1}^l \phi(b_j)$. Taking a summation of \Cref{equ:almost-k-wise} for all tuples $(c_1, c_2, \ldots, c_l)\in T$ and using triangle inequality, we have:
        $$\begin{aligned}
            \left|\Pr\left[Y_{i_j} = b_j, \forall 1\leq j\leq l\right] - \prod_{j=1}^l \phi(b_j)\right|&= \left|\sum_{(c_1, c_2, \ldots, c_l)\in T}\Pr\left[X_{i_j} = c_j, \forall 1\leq j\leq l\right] - |T| / M^l\right|\\
            &\leq |T|\cdot \epsilon M^{k-l} \leq \epsilon M^k
        \end{aligned}$$
    \end{proof}
    To bound the expectation of any term $\Pi_{j=1}^k Z_{i_j}^{n_j}$, we expand it by the definition of expectation:
    $$\begin{aligned}
        \mathbb{E}\left[\Pi_{j=1}^k Z_{i_j}^{n_j}\right] &= \sum_{(b_1, b_2, \ldots, b_l)\in \{0, 1\}^l}\Pr[Y_{i_j} = b_j, \forall 1\leq j\leq l]\cdot \prod_{j=1}^l\brac{\frac{w_{i_j}}{W}\cdot (b_j - 1/M)}^{n_j}\\
        &= \sum_{(b_1, b_2, \ldots, b_l)\in \{0, 1\}^l}\prod_{j=1}^l\phi(b_j)\cdot \brac{\frac{w_{i_j}}{W}\cdot (b_j - 1/M)}^{n_j}\\
        &+ \sum_{(b_1, b_2, \ldots, b_l)\in \{0, 1\}^l}\brac{\Pr[Y_{i_j} = b_j, \forall 1\leq j\leq l] - \prod_{j=1}^l\phi(b_j)}\cdot \prod_{j=1}^l\brac{\frac{w_{i_j}}{W}\cdot (b_j - 1/M)}^{n_j}\\
        &\leq \sum_{(b_1, b_2, \ldots, b_l)\in \{0, 1\}^l}\prod_{j=1}^l\phi(b_j)\cdot \brac{\frac{w_{i_j}}{W}\cdot (b_j - 1/M)}^{n_j}\\
        &+ \sum_{(b_1, b_2, \ldots, b_l)\in \{0, 1\}^l} \epsilon M^k\prod_{j=1}^l\brac{\frac{w_{i_j}}{W}\cdot (b_j - 1/M)}^{n_j}\\
        &\leq \sum_{(b_1, b_2, \ldots, b_l)\in \{0, 1\}^l}\prod_{j=1}^l\phi(b_j)\cdot \brac{\frac{w_{i_j}}{W}\cdot (b_j - 1/M)}^{n_j} + \epsilon (2M)^{k}
    \end{aligned}$$

    Define separately a sequence of $k$-wise independent random variables $\{\tilde{Z}_i\}_{1\leq i\leq N}$ with marginal probabilities: $$\tilde{Z}_i = \begin{cases}
        -\frac{w_i}{WM} &   \text{with probability }1-1/M\\
        \frac{w_i}{W}\cdot(1-1/M)   &   \text{with probability }1/M
    \end{cases}$$
    So, $\mathbb{E}[\tilde{Z_i}] = 0$, $|\tilde{Z_i}|\leq 1$, $\mathsf{Var}[\tilde{Z_i}]\leq 1$. Then, by the same derivation, we have:
    $$\begin{aligned}
        &\mathbb{E}\left[\brac{\sum_{i=1}^N\tilde{Z}_i}^k\right] = \sum_{(i_1, i_2, \ldots, i_l)\in [M]^l, i_j\neq i_k}\sum_{
        (n_1, n_2, \ldots, n_l), \sum_{j=1}^l n_j = k}\mathbb{E}\prod_{j=1}^l\tilde{Z}_{i_j}\\
        &= \sum_{(i_1, i_2, \ldots, i_l)\in [M]^l, i_j\neq i_k}\sum_{
        (n_1, n_2, \ldots, n_l), \sum_{j=1}^l n_j = k}\sum_{(b_1, b_2, \ldots, b_l)\in \{0, 1\}^l}\prod_{j=1}^l\phi(b_j)\cdot \brac{\frac{w_{i_j}}{W}\cdot (b_j - 1/M)}^{n_j}
    \end{aligned}$$
    % \hongyi{$\brac Z_i$ is fully independent, so we don't use the condition $k$-wise independent in Lemma C.1?} \tianyi{I should have said $k$-wise independent}
    
    Therefore, we have:
    $$\begin{aligned}
        &\mathbb{E}\left[\brac{\sum_{i=1}^N Z_i}^k\right] = \sum_{(i_1, i_2, \ldots, i_l)\in [M]^l, i_j\neq i_k}\sum_{
        (n_1, n_2, \ldots, n_l), \sum_{j=1}^l n_j = k}\mathbb{E}\prod_{j=1}^l Z_{i_j}\\
        &\leq \mathbb{E}\left[\brac{\sum_{i=1}^N\tilde{Z}_i}^k\right] + \epsilon (2M)^k N^k\leq (kN)^{k/2} + \epsilon (2M)^kN^k
    \end{aligned}$$
    The last inequality is due to \Cref{lem:kwise-conc}. Plugging in \Cref{equ:markov}, we conclude that:
    $$\Pr\left[\left|\sum_{i=1}^N w_i\cdot Y_i - \frac{1}{M}\sum_{i=1}^N w_i\right| > \frac{\delta WN}{M}\right]\leq \brac{\frac{M\sqrt{k}}{\delta \sqrt{N}}}^k + \epsilon (2M^2/ \delta)^k$$
\end{proof}

\end{document}